\documentclass[sigconf,nonacm]{acmart}

\usepackage{booktabs}
\usepackage{xspace}
\usepackage[export]{adjustbox}
\usepackage{amsmath}
\providecommand{\mathbb}[1]{\mathbf{#1}}
\usepackage{listings}
\usepackage{xcolor}
\usepackage{tikz}
\usetikzlibrary{positioning,arrows.meta,fit,backgrounds,calc}

\renewcommand\footnotetextcopyrightpermission[1]{}
\setcopyright{none}
\acmConference[Preprint]{arXiv preprint}{2026}{}
\acmBooktitle{}
\acmDOI{}
\acmISBN{}

\lstdefinestyle{sb}{
  basicstyle=\ttfamily\scriptsize,
  keywordstyle=\bfseries,
  commentstyle=\itshape\color{gray!80!black},
  numbers=none,
  breaklines=true,
  columns=fullflexible,
  showstringspaces=false,
  captionpos=b,
  frame=none,
  xleftmargin=0pt,
  aboveskip=4pt,
  belowskip=4pt,
}
\newcommand{\sysname}{SchedBlame\xspace}
\let\origunderscore\_
\newcommand{\code}[1]{\texttt{\small\def\_{\origunderscore\allowbreak}#1}}
\newtheorem{invariant}{Invariant}
\newtheorem{proposition}{Proposition}

\begin{document}

\title{SchedBlame: Who Ran While You Waited?}
\subtitle{Culprit-Attributed CPU Contention for Containers on Stock Kernels}

\author{Hao Li}
\email{haolee@didiglobal.com}
\affiliation{%
  \institution{DiDi Chuxing}
  \city{Beijing}
  \country{China}
}

\author{Tonghao Zhang}
\email{zhangtonghao@didiglobal.com}
\affiliation{%
  \institution{DiDi Chuxing}
  \city{Beijing}
  \country{China}
}

\author{Honglei Wang}
\email{wanghonglei@didiglobal.com}
\affiliation{%
  \institution{DiDi Chuxing}
  \city{Beijing}
  \country{China}
}

%% acmart builds the running head automatically ("First et al."). Uncomment and
%% edit only if the generated one is wrong.
%\renewcommand{\shortauthors}{FIRST AUTHOR et al.}

\begin{abstract}
Containers that share a machine compete for CPU. When one slows down, the
operator needs to know which co-tenant is responsible, and no deployed signal can
say. Pressure stall information, per-cgroup wait counters, and run-queue latency
histograms are all victim-side: they report that a container waited, never who it
waited for. Recovering the culprit today means patching the kernel, which does not
travel across a fleet; or enabling full scheduler tracing, which cannot stay on;
or inferring blame statistically, which misattributes exactly when several victims
are hurting at once.

\sysname is an eBPF tracer that attributes CPU contention to the cgroups that
caused it, on stock kernels, continuously. It inverts the usual accounting.
Instead of measuring how long a victim waited, it measures the CPU time every
other cgroup consumed while that victim was runnable but not running on the same
CPU. The mechanism is a per-CPU bitmap of which measured cgroups are waiting,
maintained from the kernel's own runnable counts at four scheduler hooks. Every
completed run slice carries that bitmap, so a single 16-byte record charges CPU
time to a full row of a competitor $\times$ victim blame matrix. The kernel stores
no per-pair state.

Three properties follow. Slices are self-describing, so userspace holds no waiting
state and a lost record costs measurements rather than correctness. The measured
set is reconfigured by publishing an epoch, which invalidates every classification
cache and every per-CPU bitmap in constant time while the hooks keep running. And
sampling discards slices without touching waiting state, so the sampling rate
trades variance against cost and nothing else; userspace rescales retained slices
by the inverse keep probability, which yields an unbiased estimator.

\sysname decomposes each container's per-second CPU demand into runtime, internal
contention, external contention, and bandwidth throttling, flags anomalies against
a rolling $99$\textsuperscript{th}-percentile baseline, and names the competitors
responsible. It runs on unmodified 4.18 and 5.10 kernels. In production on a
96-core host tracking 84 containers, it costs about 1\% of Redis throughput and
6\% of one core.
\end{abstract}


\maketitle

\section{Introduction}
\label{sec:intro}

Machines in a container fleet are shared on purpose. Collocating jobs is how
operators recover the utilization that strict isolation wastes~\cite{borg,heracles},
and the CPU controller's shares and bandwidth limits make the practice tolerable.
The arrangement works until a container slows down. At that point the operator has
one question, and it is a question about a different container: \emph{who took my
CPU?}

The available signals do not answer it. Pressure stall information reports the
fraction of time a cgroup's tasks were stalled on CPU~\cite{psi}. Scheduler
statistics report accumulated run-queue wait per entity~\cite{schedstats}. The CPU
controller reports how long a group was throttled against its
quota~\cite{cgroupv2,schedbwc}. Tools such as \code{runqlat} and
\code{runqslower} render run-queue latency as a histogram~\cite{bcc,runqlat}.
Every one of these measures the victim. Each says \emph{that} a container waited,
how long, and how often; none says who was running during the wait. The gap is
operational, because remediation acts on the culprit: move it, cap it, reschedule
it, or page its owner. A victim-side metric raises the alarm but does not identify
the action.

Three workarounds exist, and each gives up something a fleet cannot give up. A
kernel patch can decompose wait time precisely, but patches do not travel across
the kernel versions a fleet runs, and they make the signal a property of the
kernel build rather than of the monitoring tool. Full scheduler tracing recovers
everything at a cost that rules out leaving it on. Statistical attribution
correlates a victim's degradation against co-tenant
behaviour~\cite{cpi2,panda,causal}; it deploys easily, but it infers rather than
observes. PANDA's authors document that correlation-based antagonist
identification becomes inconsistent exactly when multiple victims coexist and
share contention~\cite{panda}, which is the common case on a busy machine.

The closest existing work does read scheduler events. Volpert et al.\ detect
noisy neighbors online from two eBPF scheduler metrics and separate that
disruption from a cgroup's self-inflicted disruption, but their instrumentation
reads only tasks of the observed cgroup, so it establishes that contention came
from outside without identifying its source~\cite{volpert}. Netflix instruments
\code{sched\_wakeup} and \code{sched\_switch}, exports per-container run-queue
latency, and tags each preemption with the cgroup that displaced the
task~\cite{netflix}. That does name a cgroup, but only the one that displaced the
victim at a single instant, and it counts preemption events rather than time. A container can be starved for a hundred milliseconds by cgroups that
never preempt it once, simply by never being picked while they run. A preemption
count also has no denominator, so it cannot be normalized against the victim's own
demand.

\paragraph{The inversion.}
\sysname starts from a different observation. Contention is not a property of the
waiting task; it is a property of the CPU. If a container is runnable on CPU~$k$
and is not running, somebody else is, and that somebody is responsible for the
delay. So instead of measuring the victim's wait, \sysname measures \emph{the
runtime of everyone else while the victim waits}. Blame becomes an observation
rather than an inference: not ``this container's latency correlates with that
container's activity'', but ``this container consumed 40\,ms of CPU~7 while that
container sat runnable on CPU~7''.

Turning that observation into a continuously running tracer is the engineering
problem. Charging every running interval to every concurrently waiting victim is a
cross product, and enumerating a cross product is exactly what a scheduler hot path
must not do. \sysname avoids it with an asymmetry. The set of cgroups that can be
\emph{blamed} is large and sparse: any cgroup with a kernel CSS ID, 12 bits on the
wire. The set of cgroups whose contention is \emph{measured} is a monitoring
choice, so it can be small, bounded, and densely numbered, which makes it fit in a
per-CPU bitmap. Each completed run slice is stamped with that bitmap as it ends,
and the resulting 16-byte record charges one competitor's CPU time to every victim
waiting behind it: an entire row of the blame matrix in one record, with no
per-pair state in the kernel.

Three properties follow, and they are what make the design deployable rather than
merely correct.

\emph{Slices are self-describing.} Each record carries the snapshot it should be
attributed against, so userspace maintains no waiting state and never replays a
transition stream. A dropped perf record costs the CPU time in those slices and
nothing else, because the next slice arrives with a fresh snapshot.

\emph{Reconfiguration is a constant-time publication.} The measured set changes as
containers come and go. \sysname pairs each cached classification with an epoch
inside one aligned 64-bit word, so publishing a new epoch invalidates every cache
entry and every per-CPU waiting bitmap at once, with no scan, no drain, and no
pause of the hooks running concurrently.

\emph{Sampling cannot buy cost with correctness.} Slices are discarded before the
snapshot is copied, but state maintenance is unconditional. Sampling governs only
which completed slices userspace observes; it never perturbs what the kernel
believes about who is waiting. Userspace rescales retained durations by the
inverse keep probability, so the estimator stays unbiased and only its variance
degrades as the operator buys back CPU.

On this substrate, \sysname decomposes each measured container's per-second CPU
demand into four terms: time it ran, time it waited behind itself, time it waited
behind others, and time its own quota forbade it to run. The third term over the
total is the external contention ratio, the quantity an operator wants. \sysname
compares it against a rolling $99$\textsuperscript{th}-percentile baseline and, on
an anomaly, reports the ranked competitors the blame matrix holds responsible.

\paragraph{Contributions.}
\begin{itemize}
  \item \textbf{Blame-the-runner attribution} (\S\ref{sec:model}): CPU contention
  formulated as competitor runtime concurrent with victim waiting, which makes
  culprit attribution a direct observation at the scheduling event instead of a
  statistical inference over aggregate metrics.
  \item \textbf{A constant-cost in-kernel realization} (\S\ref{sec:bpf}): waiting
  state read from the kernel's own runnable counts under existing runqueue locks;
  a sparse/dense identity split that fits the victim set into a per-CPU bitmap;
  epoch-tagged lock-free reconfiguration; and self-describing packed slices that
  make userspace stateless and loss-tolerant.
  \item \textbf{Sampling decoupled from state} (\S\ref{sec:analysis}): an unbiased
  inverse-probability estimator with a variance bound and a cost model, so
  overhead becomes a tuning knob with no correctness cliff.
  \item \textbf{An explicit error model} (\S\ref{sec:approx}): every approximation
  the design accepts, with the direction of its bias and the conditions under
  which it matters.
  \item \textbf{A deployed implementation} (\S\ref{sec:impl}, \S\ref{sec:eval}) on
  unmodified 4.18 and 5.10 kernels, with preliminary overhead measurements and a
  stated plan for full evaluation.
\end{itemize}

\section{Background and Problem}
\label{sec:background}

\subsection{Group scheduling in CFS}

The Linux Completely Fair Scheduler organizes runnable tasks into a hierarchy of
task groups. Each cgroup in the CPU controller corresponds to a
\code{task\_group}, which has one \code{cfs\_rq} per CPU. Two properties of this
structure are load-bearing for \sysname.

First, each per-CPU \code{cfs\_rq} maintains \code{h\_nr\_running}, the number of
runnable CFS entities in that group's hierarchy on that CPU. The entity the
scheduler has selected remains counted: picking a task does not remove it from
\code{h\_nr\_running}. A group with $n$ runnable entities on a CPU therefore has
$n-1$ waiting there if it owns the CPU, and $n$ if another group does. \sysname
never counts runnable tasks itself; it reads this number.

Second, all mutations of a CPU's runnable state happen under that CPU's runqueue
lock: enqueue, dequeue, migration source accounting, throttling, and the context
switch. A tracepoint that fires inside those critical sections observes a
consistent snapshot and cannot race a concurrent update on the same CPU.
\sysname's correctness arguments (\S\ref{sec:bpf:order}) are all arguments about
where a hook sits relative to this lock and to the kernel's own counter updates.

The CPU controller also enforces bandwidth limits. When a group exhausts its quota
within a period, \code{throttle\_cfs\_rq()} detaches its \code{cfs\_rq} from the
schedulable hierarchy, and \code{unthrottle\_cfs\_rq()} reattaches it when the
quota refills~\cite{schedbwc}. Throttling is not contention: nobody took the CPU
from the group, its own limit did. A tool that conflates the two blames innocent
neighbors for a configuration decision, so \sysname accounts throttling
separately.

\subsection{What deployed signals report}

Table~\ref{tab:signals} lists the CPU contention signals available on an
unmodified kernel. They differ in granularity and cost and agree in one respect:
all describe the victim.

\begin{table}[t]
\centering
\small
\caption{CPU contention signals on stock kernels. All are victim-side.}
\label{tab:signals}
\begin{adjustbox}{max width=\columnwidth}
\begin{tabular}{@{}lll@{}}
\toprule
Signal & Reports & Names culprit? \\
\midrule
\code{cpu.pressure} (PSI) & stall fraction per cgroup & no \\
\code{schedstat} \code{wait\_sum} & accumulated wait per entity & no \\
\code{cpu.stat} throttling & quota-induced stall & n/a \\
\code{cpuacct.usage} & consumed CPU time & no \\
\code{runqlat}/\code{runqslower} & wait-latency distribution & no \\
\midrule
\sysname & wait time \emph{per competitor} & yes \\
\bottomrule
\end{tabular}
\end{adjustbox}
\end{table}

This is not an oversight but a consequence of where the measurement happens. A
wait counter is incremented against the entity that is waiting. The identity of
whoever occupies the CPU is available at that moment, since it is \code{rq->curr},
but no accounting path records it: recording it means recording a relationship
rather than a scalar, and a relationship has no natural home in a per-cgroup
counter file.

\subsection{Why the obvious alternatives do not deploy}

\paragraph{Patch the kernel.}
The scheduler can decompose wait time internally, separating waiting caused by
siblings within a group from waiting caused by outside competition. Such counters
are exact and nearly free, because the kernel is already at the right place with
the right locks held. They are also unavailable in practice. A fleet runs several
kernel versions at once, patches must be forward-ported indefinitely, and the
signal becomes a property of the kernel build rather than of the tool. Even where
such a counter exists it aggregates: it reports how much external competition a
group suffered, not which group supplied it.

\paragraph{Trace every switch.}
Emitting a record per context switch and reconstructing the schedule offline
recovers everything, culprits included. On a busy machine, context switches arrive
fast enough that per-event userspace delivery consumes a significant fraction of a
core, and the resulting volume cannot be retained continuously. Tools in this
family are diagnostic instruments, invoked once a problem is known. Contention,
however, is often bursty and unreproducible: by the time an operator attaches a
tracer, the culprit has finished.

\paragraph{Infer blame statistically.}
CPI\textsuperscript{2} detects CPI outliers and nominates likely
perpetrators~\cite{cpi2}. PANDA ranks antagonists by correlation against a
machine-level metric, using global historical data to suppress sampling
noise~\cite{panda}. Causal-inference frameworks apply Granger causality to
resource time series~\cite{causal}. These methods deploy easily and work at
cluster scale, but they are indirect, and their failure mode is intrinsic: when
several victims coexist and share contention, correlation-based attribution
becomes inconsistent~\cite{panda}. Two containers that ramp up together are
indistinguishable to a correlational method, and a busy machine supplies many such
pairs.

\paragraph{Instrument the scheduler.}
Two efforts read scheduler events directly, and they mark the state of the art
this work extends. Volpert et al.\ hook \code{sched\_wakeup} and
\code{sched\_switch} and derive two per-second, workload-agnostic metrics for an
observed cgroup: average process scheduling latency (\emph{PSL}), and the average
number of switches whose incoming task is the idle task (\emph{PSP}), which they
read as an indicator of quota throttling. A four-quadrant decision matrix
interprets the pair: low PSL and low PSP means undisturbed, high PSL with low PSP
indicates a noisy neighbor, and high PSL with high PSP indicates disruption
originating inside the group~\cite{volpert}. The scheme is online, needs no
workload profile, and separates external contention from self-inflicted
disruption, which is a real advance over the counters above.

Its instrumentation, however, only ever examines the observed cgroup. At each
switch it either counts an idle-task preemption or, when the incoming task belongs
to the observed group, closes out that task's wait; the competing cgroup is never
read. The output therefore says that contention came from outside, not which
cgroup supplied it. Netflix's tracer does record the other party, tagging each
preemption with the cgroup of the outgoing task and classifying it as
same-container, other-container, or system service~\cite{netflix}.

Preemption tagging is the closest deployed approach to naming a culprit, and its
limitation is structural: a tag captures a single instant and counts events rather
than time. A container starved for a hundred milliseconds because it was never
picked, while three other groups ran in turn, produces no preemption attributable
to any of them. Neither metric has a denominator in the victim's own demand, so a
busy container and a starved one are hard to tell apart.

\subsection{Requirements}

These gaps define \sysname's requirements.

\begin{enumerate}
  \item \textbf{Culprit-attributed.} Name the cgroups responsible and quantify
  each one's contribution in time.
  \item \textbf{Stock kernel.} No kernel patch. Only tracepoints, kprobes, and BPF
  facilities present in production kernels, including old ones.
  \item \textbf{Continuously on.} Overhead low enough to leave running, so that
  bursty contention is captured when it happens.
  \item \textbf{Bounded state.} Memory and per-event work independent of task
  count and of the workload's switching behaviour.
  \item \textbf{Degrade, do not corrupt.} Under buffer loss or reconfiguration,
  lose measurements, never accumulate persistent error.
  \item \textbf{Normalized output.} Express contention relative to the victim's own
  demand, so thresholds are comparable across containers.
\end{enumerate}

\section{Blame the Runner}
\label{sec:model}

\subsection{The formulation}

Consider one CPU $k$. At every instant exactly one cgroup owns it, and some set of
other cgroups have runnable tasks queued behind. Let $\rho_k(\tau)$ be the cgroup
running on CPU $k$ at time $\tau$, and let $w_{k,t}(\tau) \in \{0,1\}$ indicate
that measured cgroup $t$ has at least one runnable task on CPU $k$ that is not
running at time $\tau$.

\sysname's blame quantity is the time integral of their coincidence. For a
competitor cgroup $c$ and a measured cgroup (a \emph{target}) $t$,
\begin{equation}
\label{eq:blame}
B(c,t) \;=\; \sum_{k}\int \mathbf{1}[\rho_k(\tau) = c]\; w_{k,t}(\tau)\; d\tau ,
\qquad c \neq t .
\end{equation}

$B(c,t)$ has units of CPU-nanoseconds and reads directly: how much CPU time $c$
consumed while $t$ was waiting for a CPU it could have used. Summing over
competitors gives $t$'s external contention $E_t = \sum_{c \neq t} B(c,t)$, and the
individual terms rank the culprits.

The same integral with $c = t$ measures something different: the target running on
one CPU while another of its own tasks waits there. That is self-inflicted
queueing, not interference. \sysname keeps it as a separate term $I_t$,
\emph{internal contention}, both because operators must not be paged for their own
parallelism and because $I_t$ is a useful control. A target with high $I_t$ and low
$E_t$ is oversubscribed by itself.

\subsection{Demand decomposition}

Blame is actionable only relative to how much CPU the victim wanted. \sysname
decomposes each target's demand over an evaluation interval into four terms:
\begin{align*}
R_t &= \text{CPU time the target ran} \\
I_t &= \text{internal contention (waited behind itself)} \\
E_t &= \text{external contention (waited behind others)} \\
T_t &= \text{time directly throttled against its own quota}
\end{align*}
with total wait and total demand
\begin{equation}
\label{eq:demand}
W_t = I_t + E_t + T_t, \qquad D_t = R_t + W_t ,
\end{equation}
and the reported quantity
\begin{equation}
\label{eq:ratio}
r_t \;=\; \frac{E_t}{D_t} \;\in\; [0,1] .
\end{equation}

Every unit of target demand falls into exactly one term: the target ran, waited
behind itself, waited behind someone else, or was forbidden to run by its own
bandwidth limit. So $r_t$ is the fraction of \emph{wanted} CPU that external
competition took. Unlike a raw wait time or a preemption count, it is comparable
across containers of different sizes and activity levels.

\paragraph{One demand unit per CPU.}
The $w_{k,t}$ in Eq.~\ref{eq:blame} is an indicator, not a count. If a target has
five runnable tasks queued on a CPU, \sysname charges the interval once, not five
times. This is a deliberate departure from task-weighted kernel wait accounting,
which multiplies elapsed time by the number of waiting tasks. The indicator
formulation is what lets a target's state on a CPU occupy a single bit, which in
turn is what makes the per-CPU snapshot small enough to travel inside every slice.
The cost is that $r_t$ measures how much of the target's opportunity to run was
taken, rather than how much aggregate task time was lost. For a target whose
runnable-task count is stable the two agree up to a constant that cancels in the
ratio; when parallelism changes sharply within an interval they diverge.
\S\ref{sec:approx} states this and the other approximations precisely.

\subsection{Why this is measurable cheaply}

Eq.~\ref{eq:blame} is a cross product over (competitor, target) pairs, and
computing a cross product on the scheduler's hot path is what a continuously
running tracer cannot afford. Two observations make it tractable.

\paragraph{The two sides are asymmetric.}
Anything can be a culprit, so the competitor space must be large: \sysname
identifies competitors by the kernel's sparse CSS ID, 12 bits on the wire,
covering 4096 concurrent task groups. But which cgroups are \emph{measured} is a
monitoring choice, and that set can be small and bounded. \sysname gives each
target a \emph{dense} identifier that is directly a bit position, so ``which
targets are waiting on CPU $k$'' is a machine word or two.

\paragraph{Blame is a property of the ending boundary.}
The inner integral of Eq.~\ref{eq:blame} is piecewise constant, because $\rho_k$
changes only at a context switch. A context switch therefore delimits a
\emph{slice}: an interval during which exactly one cgroup ran on one CPU. If the
tracer stamps the slice with that CPU's waiting bitmap at the moment it ends, the
single record
\[
\langle\, \text{competitor CSS ID},\; \text{duration},\; \text{waiting bitmap} \,\rangle
\]
carries everything needed to charge that competitor against every waiting target
at once. One record fills an entire row of the blame matrix. The kernel enumerates
no pairs, allocates no per-pair state, and does no work proportional to the number
of targets; it copies a fixed-size bitmap it was already maintaining.

Figure~\ref{fig:inversion} shows the resulting picture on one CPU.

\begin{figure}[t]
\centering
\begin{adjustbox}{max width=\columnwidth}
\begin{tikzpicture}[x=1cm,y=1cm,font=\scriptsize]
  \draw[->] (0.4,0.2) -- (7.4,0.2) node[right] {time};
  % running slices on CPU k
  \fill[blue!18]   (0.6,0.45) rectangle (2.3,1.1);
  \fill[orange!25] (2.3,0.45) rectangle (4.2,1.1);
  \fill[green!20]  (4.2,0.45) rectangle (5.4,1.1);
  \fill[blue!18]   (5.4,0.45) rectangle (6.9,1.1);
  \draw (0.6,0.45) rectangle (2.3,1.1);
  \draw (2.3,0.45) rectangle (4.2,1.1);
  \draw (4.2,0.45) rectangle (5.4,1.1);
  \draw (5.4,0.45) rectangle (6.9,1.1);
  \node at (1.45,0.78) {$c_1$};
  \node at (3.25,0.78) {$c_2$};
  \node at (4.80,0.78) {$c_3$};
  \node at (6.15,0.78) {$c_1$};
  \node[left] at (0.55,0.78) {CPU $k$};
  % waiting interval of target t
  \draw[very thick,red!70!black] (1.4,1.6) -- (5.9,1.6);
  \draw[red!70!black] (1.4,1.45) -- (1.4,1.75);
  \draw[red!70!black] (5.9,1.45) -- (5.9,1.75);
  \node[above] at (3.65,1.65) {target $t$ runnable, not running};
  % charge arrows
  \foreach \x in {1.85,3.25,4.80,5.65} {
    \draw[-{Latex[length=1.4mm]},gray!70] (\x,1.52) -- (\x,1.16);
  }
\end{tikzpicture}
\end{adjustbox}
\caption{Blame the runner. While target $t$ is runnable but not running on
CPU~$k$, whoever occupies the CPU is charged for that time. Every competitor that
ran during the wait is charged, not only the one that displaced $t$.}
\label{fig:inversion}
\Description{A timeline of one CPU showing four consecutive run slices belonging to competitors c1, c2, c3 and c1 again. A bar above marks the interval in which target t is runnable but not running, and arrows charge each overlapping competitor slice to t.}
\end{figure}
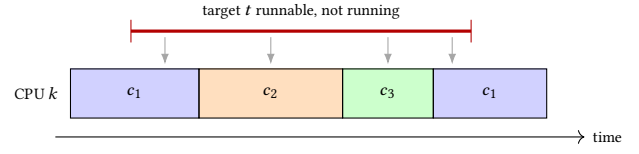

\section{In-Kernel Design}
\label{sec:bpf}

\begin{figure}[t]
\centering
\begin{adjustbox}{max width=\columnwidth}
\begin{tikzpicture}[font=\scriptsize,node distance=3mm,
  box/.style={draw,rounded corners=1pt,align=center,inner sep=3pt,minimum height=6mm},
  k/.style={box,fill=blue!8},
  u/.style={box,fill=orange!10}]
  \node[k,minimum width=2.1cm] (hooks) {scheduler hooks\\ \code{switch}, \code{wakeup},\\ \code{migrate}, \code{unthrottle}};
  \node[k,right=6mm of hooks,minimum width=1.9cm] (bitmap) {per-CPU\\waiting bitmap};
  \node[k,below=5mm of bitmap,minimum width=1.9cm] (batch) {per-CPU\\slice batch};
  \node[k,left=6mm of batch,minimum width=2.1cm] (cache) {epoch-tagged\\ \code{css\_id}$\rightarrow$dense\\cache};
  \node[u,right=6mm of bitmap,minimum width=1.7cm] (perf) {perf ring\\(per CPU)};
  \node[u,below=5mm of perf,minimum width=1.7cm] (runner) {runner\\(single owner)};
  \node[u,below=5mm of runner,minimum width=1.7cm] (matrix) {charge matrix\\ $c \times t$};
  \draw[-{Latex[length=1.2mm]}] (hooks) -- (bitmap);
  \draw[-{Latex[length=1.2mm]}] (hooks) -- (cache);
  \draw[-{Latex[length=1.2mm]}] (bitmap) -- node[right,pos=0.4]{stamp} (batch);
  \draw[-{Latex[length=1.2mm]}] (batch.east) -- ++(0.35,0) |- (perf.west);
  \draw[-{Latex[length=1.2mm]}] (perf) -- (runner);
  \draw[-{Latex[length=1.2mm]}] (runner) -- (matrix);
  \draw[-{Latex[length=1.2mm]},dashed] (runner.west) -- ++(-0.5,0) |- node[pos=0.25,left,align=right]{epoch\\publish} (cache.east);
  \begin{scope}[on background layer]
    \node[draw,dashed,gray,rounded corners,fit=(hooks)(bitmap)(batch)(cache),
      label={[gray,font=\scriptsize]below:kernel (BPF)}] {};
    \node[draw,dashed,gray,rounded corners,fit=(perf)(runner)(matrix),
      label={[gray,font=\scriptsize]below:userspace}] {};
  \end{scope}
\end{tikzpicture}
\end{adjustbox}
\caption{\sysname architecture. The only userspace-to-kernel write is the target
epoch; everything else flows outward.}
\label{fig:arch}
\Description{Block diagram. In the kernel, scheduler hooks update a per-CPU waiting bitmap and an epoch-tagged classification cache; the bitmap stamps slices appended to a per-CPU batch. Batches flow through a per-CPU perf ring to a single userspace runner, which builds the competitor-by-target charge matrix and publishes target epochs back to the cache.}
\end{figure}
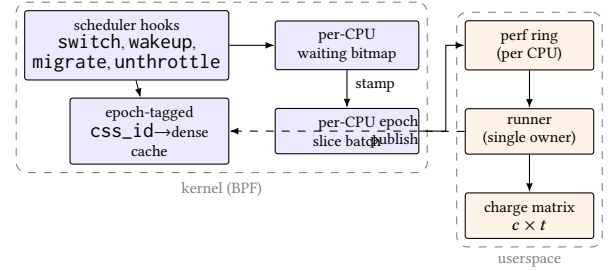

\subsection{Per-CPU waiting state}

The tracer maintains one piece of state: for each CPU, the set of targets waiting
there.
\[
\code{waiting\_bitmap[cpu]} \;\subseteq\; \{0,\dots,M-1\}
\]
$M$ is the compiled target capacity. The representation is a 20-bit base field
plus a compile-time number of 64-bit extension words,
\[
M = 20 + 64 \cdot \textsc{extra},
\]
because the base field shares a 64-bit word with the competitor identifier and the
slice duration (\S\ref{sec:bpf:slice}). The default $\textsc{extra}=1$ gives
$M=84$ targets in a 16-byte slice. The dense identifier is a byte, so
$\textsc{extra} \le 3$ and $M \le 212$.

Two properties of this state drive the rest of the design. It is
$O(\text{CPUs})$, independent of how many tasks or containers are on the machine:
there is no per-task map, no per-cgroup runnable counter, and no hash lookup keyed
by anything that grows with the workload. And it is authoritative but not durable.
It is derived from the kernel's own counters at every relevant hook, so it can be
discarded and rebuilt at any time, which is what makes the reconfiguration scheme
of \S\ref{sec:bpf:epoch} safe.

\subsection{Reading waiting state instead of tracking it}
\label{sec:bpf:order}

A tracer could maintain waiting state by counting: increment on enqueue, decrement
on dequeue, infer waiting from the difference. That requires seeing every
transition, which makes it fragile against a missed or reordered event and ties
correctness to the event stream. \sysname does not count. At each hook it reads
\code{h\_nr\_running}, the kernel's own runnable count for that group on that CPU,
and recomputes the bit. A lost update is repaired by the next observation instead
of corrupting an accumulator.

What makes this sound is where the hooks sit. Each argument below is an invariant
about kernel ordering, and each holds identically on 4.18 and 5.10.

\begin{invariant}[Switch]
\label{inv:switch}
\code{trace\_sched\_switch} fires after the scheduler has updated runnable counts
for the outgoing task and selected the incoming one, but before the physical
context switch, with the runqueue lock held. A selected task remains counted in
\code{h\_nr\_running}.
\end{invariant}

Two consequences follow. First, on hook entry the waiting bitmap still describes
the interval in which \code{prev} ran, so the completed slice is stamped before any
update is applied. Second, after the switch, ``\code{prev}'s group still has a
runnable entity here'' means one is waiting, since \code{prev} has left the CPU;
and ``\code{next}'s group has at least two'' means a sibling remains queued behind
\code{next}. The hook therefore stamps the completed slice against the old bitmap,
sets \code{prev}'s bit from $\code{h\_nr\_running} \ge 1$, and sets \code{next}'s
bit from $\code{h\_nr\_running} \ge 2$, in that order. Applying \code{next} second
makes same-cgroup switches settle correctly.

\begin{invariant}[Wakeup]
\label{inv:wakeup}
\code{trace\_sched\_wakeup} and \code{trace\_sched\_wakeup\_new} fire after the
task has been enqueued on its destination \code{cfs\_rq} and the destination count
incremented, under the destination runqueue lock.
\end{invariant}

That lock also serializes \code{sched\_switch} on the destination CPU, so the hook
can safely combine the freshly incremented count with the tracer's record of which
group is running there: if that group owns the CPU, one entity runs and $n-1$
wait; otherwise all $n$ wait. \sysname does not use \code{sched\_waking}, which
fires before enqueue and cannot see the updated count. It also ignores
self-wakeups, where the task is already current and merely has its state restored
to \code{TASK\_RUNNING}: no entity is enqueued, so the bitmap must not move.

\begin{invariant}[Migration]
\label{inv:migrate}
For a queued runnable CFS task, \code{trace\_sched\_migrate\_task} fires inside
\code{set\_task\_cpu()} while the \emph{source} runqueue lock is held, after the
task has been dequeued from the source and before its CPU field changes. The task
still points at the source \code{cfs\_rq}.
\end{invariant}

The hook can therefore recompute the source bit exactly. Without it, a target's
waiting bit would stay set on a CPU its last waiting task had left, and every
later slice on that CPU would be blamed for a victim that is no longer there. The
destination bit is deliberately not written: the destination runqueue lock is not
held, and a remote read-modify-write of another CPU's bitmap would race. The
destination is reconstructed by its own next switch or wakeup. The same tracepoint
also fires during wakeup placement, where only \code{task->pi\_lock} may be held,
but there the task is \code{TASK\_WAKING} and is filtered out before any runqueue
state is read.

\begin{invariant}[Throttling]
\label{inv:throttle}
\code{throttle\_cfs\_rq()} and \code{unthrottle\_cfs\_rq()} execute with the
affected runqueue lock held. Throttling of the current group completes before the
resulting \code{sched\_switch}.
\end{invariant}

Throttled time must not be charged to neighbors. Because throttling completes
before the switch, the switch hook's post-snapshot update already sees
\code{cfs\_rq->throttled} and clears the target's bit. The completed pre-throttle
slice keeps its original state, and later slices on that CPU exclude the throttled
group, with no pending state and no separate bitmap pass. The enqueue-time throttle
path is covered symmetrically by the wakeup hooks.

Unthrottling is the one case that needs a return probe, for two reasons: an
unthrottle attempt can return with the runqueue still throttled, and the operation
may run on a CPU other than the affected one. \sysname stores the argument at
entry, keyed by the hook's CPU, which is safe because the held runqueue lock
disables preemption. On return it derives the affected CPU through
\code{cfs\_rq->rq}, rebuilds the waiting bit from \code{h\_nr\_running}, and emits
the elapsed throttle duration computed from the same two fields the kernel itself
uses.

\subsection{Sparse competitors, dense targets}
\label{sec:bpf:identity}

Userspace and the kernel name cgroups differently, and reconciling them without a
lookup on the hot path is a small but essential piece of the design.

Userspace knows a container's CPU cgroup by the address of its CSS, which is also
the \code{task\_group} address because \code{css} is that struct's first member;
call it the \emph{cgid}. Userspace does not know the kernel's numeric
\code{css.id}. BPF, given a task, can read both. So userspace publishes the
assignment it owns in the only currency it can name,
\[
\code{target\_cgid\_to\_dense}[\,\text{cgid}\,] = \text{dense id},
\]
and BPF converts it into the currency it needs, an array indexed by
\code{css\_id}. The conversion is lazy. The first hook to observe a CSS ID
performs the map lookup and caches the result; every later observation is an array
load with no helper call. Since \code{css\_id} is bounded at 4096, the cache is a
flat 32\,KiB array.

\subsection{Reconfiguration as an epoch publication}
\label{sec:bpf:epoch}

Containers start and stop, so the target set changes while the tracer runs. The
usual approaches either scan state proportional to the ID space or quiesce the
hooks, and the hooks here run in the scheduler's critical path. \sysname does
neither.

Each cache entry is one aligned 64-bit word holding an epoch in the high 32 bits
and a dense ID in the low 8:
\[
\code{cache}[\text{css\_id}] = (\text{epoch} \ll 32)\;|\;\text{dense id}.
\]
Because the pair is loaded and stored as a single word, no concurrent hook can
observe an epoch from one publication paired with a dense ID from another. A
lookup is valid only if its epoch matches the current one; otherwise the entry is
rebuilt from the map. Publishing a new assignment is then three steps:

\begin{enumerate}
  \item rebuild \code{target\_cgid\_to\_dense} completely;
  \item update userspace's own dense-slot metadata;
  \item increment the single userspace-written global,
        \code{target\_css\_id\_epoch}.
\end{enumerate}

The increment invalidates every cache entry at once. It also invalidates every
per-CPU waiting bitmap: each CPU records which epoch its bitmap belongs to, and on
its first hook in a new epoch it clears the bitmap and adopts the new epoch.
Nothing is scanned, nothing is drained, no hook is paused, and a vacated dense slot
can be reused immediately without waiting for in-flight slices to retire. During
step~1 a hook may briefly resolve against a partially rebuilt map, but that result
is cached under the \emph{old} epoch and step~3 discards it.

This is where the derived, non-durable nature of the waiting state pays off. An
epoch bump destroys the tracer's entire notion of who is waiting on every CPU, and
that is acceptable precisely because the state is re-read from the kernel rather
than accumulated: the next switch, wakeup, or unthrottle on each CPU rebuilds it.
Reconfiguration costs a brief under-report, bounded by the time to the next
scheduling event on each CPU, and buys a path with no locks, no scans, and no
coordination with the hot path.

Epoch $0$ means unpublished, and hooks are not attached while it is current, so the
tracer never sees traffic against an undefined assignment. The epoch does not
wrap: exhausting a 32-bit counter stops the tracer, as does any failure to write
the map or the epoch. \sysname never continues on a partially published
assignment, because a wrong assignment silently misattributes blame, which is the
one failure mode a blame-assigning tool must not have.

\subsection{Self-describing slices}
\label{sec:bpf:slice}

At each switch the completed slice is packed into one 64-bit word plus the bitmap
extension words:

\begin{center}
\small
\begin{tabular}{@{}llr@{}}
\toprule
bits & field & width \\
\midrule
$0$--$11$  & competitor CSS ID & 12 \\
$12$--$31$ & waiting targets $0$--$19$ & 20 \\
$32$--$63$ & slice duration (ns) & 32 \\
\midrule
\multicolumn{2}{@{}l}{$+\;\textsc{extra} \times 64$ bits: targets $20\dots M-1$} & \\
\bottomrule
\end{tabular}
\end{center}

The layout makes the common configuration 16 bytes: one word for identity,
duration, and the first 20 targets, one word for the remaining 64. The 32-bit
duration saturates at about 4.29\,s, counted rather than wrapped.

The record is complete on its own. It references no tracer state, no event
sequence, and no prior record. Userspace can attribute it in isolation, holds no
waiting state, and never replays transitions. A lost perf record or a failed batch
submission therefore costs exactly the CPU time in those slices: the next slice
arrives with a fresh snapshot and attribution is correct again immediately. This
is what lets \sysname drop data under pressure without accumulating error, and it
is not available to any design that reconstructs waiting from an event stream.

The snapshot is taken at the slice's \emph{ending} boundary and applied to the
whole slice, so a wakeup or migration part-way through another cgroup's slice is
represented by the state at its end (\S\ref{sec:approx}).

\subsection{Sampling that cannot corrupt}
\label{sec:bpf:sampling}

Slice records dominate volume, so \sysname samples them: a completed slice is
discarded with probability $1-p$, and $p$ is a runtime knob.

The discipline is that the sampling decision happens before the snapshot is copied
and gates no state update. The switch hook always performs its \code{prev} and
\code{next} bitmap updates, whether or not the completed slice was retained, and
wakeup, migration, and unthrottle updates are likewise unconditional. Sampling
changes only which completed slices userspace observes.

This separation is unusual in sampled tracers and it matters. Where the sampled
event stream is also the state stream, as in any design that infers waiting from
observed transitions, dropping events corrupts the model. The sampling rate is
then bounded by correctness and is typically biased: drop when hot, keep when
quiet. Here it is bounded only by statistical precision, so it can be uniform and
therefore unbiased, and userspace simply divides each retained duration by $p$
(\S\ref{sec:analysis}). Placing the decision before the copy also means a discarded
slice pays neither the bitmap copy nor the transport.

\subsection{Transport}

Per-slice perf submission would spend more time in framing than in measurement, so
each CPU accumulates slices into a batch, up to 128 by default, emitted when full
or when the batch exceeds an age bound checked at the next switch on that CPU. A
full batch is a single 2072-byte perf record carrying 128 slices, amortizing one
helper call and one userspace dequeue across all of them.

Old verifiers constrain the encoding. The 4.18 verifier requires the size argument
of \code{bpf\_perf\_event\_output} to be a compile-time constant, so partial
batches are emitted at the smallest fixed capacity in $1,2,4,\dots,128$ that
contains the valid count, and userspace reads only the valid prefix. The batch
header carries a magic value and the compiled bitmap width. Userspace rejects any
batch whose width disagrees with its own and verifies the map value size before
attaching, so a mismatched kernel object and userspace binary fail loudly rather
than misattributing.

Two record types bypass batching. \emph{Identity} records publish the sparse
\code{css\_id} to cgid relationship the first time a CSS ID is seen.
\emph{Throttle} records report a completed direct-throttle interval with its
already-resolved dense target ID. Both are unsampled: identity because a slice
whose competitor is unknown is unattributable and discarded, throttle because $T$
appears in the denominator of every ratio. Neither flushes a pending slice batch,
since identity records have no ordering dependency and throttle durations are
interval quantities rather than points in the slice stream.

\section{Userspace Attribution}
\label{sec:userspace}

\subsection{A single owner}

One runner owns all mutable attribution state: target assignments, identity
tables, the charge matrix, ratio histories, and epoch publication. Perf draining,
the one-second evaluation, and target synchronization run serially inside it.
Container lifecycle callbacks, which may fire concurrently, only raise a
coalescing signal and never touch state. Anomaly reports are handed to an upload
worker as immutable documents, so the worker never reads mutable state.

This is not just defensive engineering. The blame matrix is written on every slice
and read at interval boundaries, and any locking scheme fine-grained enough to be
cheap would admit reports assembled from inconsistent rows. A report that names
the wrong culprit is worse than no report. Serializing the pipeline in one owner
makes each one-second report a consistent snapshot by construction, and costs
nothing because the work is dominated by decoding.

\subsection{Attributing a slice}

Decoding is arithmetic on the packed word; attribution is a walk over the set bits
of its snapshot.

\begin{enumerate}
  \item Unpack the competitor CSS ID and duration. Discard the slice if that CSS
  ID has no known identity yet.
  \item Scale the duration: $\hat{d} = d / p$.
  \item If the competitor is itself a target, add $\hat{d}$ to its runtime $R$.
  If its own bit is set in the snapshot, meaning it ran while another of its own
  tasks waited on that CPU, also add $\hat{d}$ to its internal contention $I$.
  \item For every \emph{other} target bit set in the snapshot, add $\hat{d}$ to the
  charge-matrix cell (competitor, target) and to that target's external contention
  $E$.
\end{enumerate}

Step~4 iterates set bits with a trailing-zero-count loop over the base field and
each extension word, so its cost is proportional to the number of targets actually
waiting, not to the target capacity. The charge matrix is one flat
$4096 \times M$ array of 64-bit counters, allocated once and 2.6\,MiB at $M=84$,
so a cell update is a single contiguous indexed add with no allocation and no
hashing.

Userspace keeps no per-cgroup runnable counts, reconstructs no waiting
transitions, and holds no state a lost record could invalidate. A throttle record
is consumed purely as a completed duration added to that target's $T$.

\subsection{One-second evaluation}

Every second, for each active target, the runner forms $W = I+E+T$ and $D = R+W$
(Eq.~\ref{eq:demand}). If $D = 0$ the interval is skipped, since no demand means no
contention to report. Otherwise it computes the external contention ratio
$r = E/D$.

Anomaly detection is per-target and baseline-relative. A fixed threshold cannot
work: a batch container may sit at 40\% external contention by design, while an
interactive service at 5\% is already in trouble. \sysname keeps each target's last
600 valid ratios and, once at least 60 samples exist, reports when
\[
r \;>\; K \cdot P_{99}(\text{history}),
\]
where $K$ defaults to $1.0$. The current sample is compared against the baseline
before being appended to it, so a sustained shift raises the baseline and stops
paging: \sysname reports departures from a container's normal contention, not its
absolute level. A historical $P_{99}$ of zero counts as no baseline. The sample
joins the history but nothing is reported, because a container that has never
experienced contention offers no scale against which to judge the first
occurrence.

A report carries the full decomposition ($R$, $I$, $E$, $T$, the ratio, the
baseline, the threshold) and, from the charge matrix, the ranked competitors and
the CPU-nanoseconds each is responsible for. That is the payload an operator acts
on. All accumulators are then cleared for the next interval.

Reports are queued to a bounded channel and uploaded asynchronously. A full queue
drops the new report and increments a counter rather than blocking the runner:
under a contention storm, the measurement pipeline must not become a second source
of contention.

\subsection{Target set maintenance}

Targets are chosen from the containers on the machine, filtered by scope and QoS
class, and assigned to dense slots. Selection is re-evaluated when a cgroup
lifecycle notification arrives, and unconditionally every 60 seconds in case a
notification is lost.

Re-selection preserves stability. A target is identified by both its container ID
and its cgid, so a surviving eligible target keeps its dense slot and refreshes its
metadata. Only genuinely vacant slots are refilled at random from the remaining
candidates: slots whose container disappeared, was recreated with a new cgid, or
stopped matching the filters. Stability matters because a dense slot carries 600
samples of ratio history, and reshuffling assignments would discard the baselines
anomaly detection depends on. When a slot is reused, its accumulators and history
are cleared, so a new occupant never inherits its predecessor's baseline.

If container discovery fails at runtime, the current assignment is retained and
retried shortly. \sysname never publishes an empty replacement, because an empty
target set silently stops all measurement while the tracer appears healthy.

\section{Estimator and Cost Model}
\label{sec:analysis}

\subsection{The sampling estimator}

Fix an evaluation interval and a quantity to measure, say the external contention
$E_t$ of target $t$, or one cell $B(c,t)$ of the charge matrix. Let
$d_1,\dots,d_n$ be the durations of the completed slices that contribute to it:
those whose ending snapshot has $t$'s bit set and whose competitor is $c$. The
true value is $\sum_i d_i$. \sysname observes each slice independently with
probability $p$ and reports
\begin{equation}
\hat{V} \;=\; \sum_{i=1}^{n} \frac{d_i}{p}\, X_i , \qquad X_i \sim \text{Bernoulli}(p) ,
\end{equation}
the Horvitz--Thompson estimator for this design~\cite{ht1952}.

\begin{proposition}[Unbiasedness]
\label{prop:unbiased}
$\mathbb{E}[\hat{V}] = \sum_i d_i$.
\end{proposition}
\begin{proof}
$\mathbb{E}[X_i] = p$, so $\mathbb{E}[\hat{V}] = \sum_i (d_i/p)\,p = \sum_i d_i$.
\end{proof}

\begin{proposition}[Variance]
\label{prop:variance}
$\operatorname{Var}[\hat{V}] = \frac{1-p}{p}\sum_i d_i^2$, and if all slices have
equal duration the relative standard error is $\sqrt{(1-p)/(pn)}$.
\end{proposition}
\begin{proof}
The $X_i$ are independent with $\operatorname{Var}[X_i] = p(1-p)$, so
$\operatorname{Var}[\hat{V}] = \sum_i (d_i/p)^2 p(1-p) = \frac{1-p}{p}\sum_i d_i^2$.
With $d_i = d$, $\sqrt{\operatorname{Var}}/\!\sum_i d_i =
\sqrt{\frac{1-p}{p} n d^2}\,/\,(nd) = \sqrt{(1-p)/(pn)}$.
\end{proof}

Proposition~\ref{prop:variance} says accuracy is governed by the number of
contributing slices, not by the sampling rate alone. Contended intervals are
exactly the intervals with many slices, so precision is highest where it matters:
a target waiting across several CPUs on a machine switching a few thousand times
per second per CPU accumulates thousands of contributing slices in one second, and
at $p = 0.1$ with $n = 1000$ the relative standard error is about 9.5\%. The
opposite regime needs care. A target that waits only briefly contributes few
slices, and a one-second interval may contain none at all, so aggressive sampling
degrades the tail of small contention events first. \S\ref{sec:approx} records
this.

Two caveats bound Proposition~\ref{prop:unbiased}. It covers \emph{configured}
sampling only: slices lost to a full perf ring or a failed batch submission are not
compensated, and such loss correlates with load, so sustained loss biases the
estimate downward. \sysname exports the counters needed to detect this
(\S\ref{sec:eval}) rather than absorbing it silently. And the $d_i$ are themselves
subject to the ending-boundary approximation (\S\ref{sec:approx}), so
unbiasedness is with respect to the sampling design, not to the underlying
scheduler.

\subsection{Cost model}

\paragraph{Per context switch.}
The hook does a constant amount of work: read the epoch, compare it against the
CPU's bitmap epoch, resolve two task groups through the cache (an array load each,
with a map lookup only on the first observation of a CSS ID in the current epoch),
update at most two bits, and with probability $p$ copy $1+\textsc{extra}$ words
into the per-CPU batch. There is no loop over targets, no per-task lookup, and no
allocation. Perf submission happens once per 128 retained slices. The work is
independent of the number of targets $M$, of the number of containers, and of the
number of runnable tasks.

\paragraph{Per slice in userspace.}
Decoding is fixed arithmetic. Attribution iterates the set bits of the snapshot,
costing $O(|\{t : w_t = 1\}|)$, the number of targets actually waiting on that
CPU, rather than $O(M)$.

\paragraph{Memory.}
Kernel-side state is fixed by compile-time bounds, not by the workload: about
32\,KiB for the classification cache, 32\,KiB for the identity table, 8\,KiB of
per-CPU waiting bitmaps, a few kilobytes of per-CPU scalars, and a 2056-byte batch
slot per CPU. Userspace is dominated by the charge matrix at
$4096 \times M \times 8$ bytes, or 2.6\,MiB at $M = 84$, plus 600 ratio samples per
target. Nothing grows with task count, container count, or uptime.

\paragraph{Transport volume.}
With $C$ CPUs switching $f$ times per second and keep probability $p$, retained
slices arrive at $pfC$ per second, at $8(1+\textsc{extra})$ bytes each and $1/128$
of that in perf records. At $C = 64$, $f = 2000$, $p = 1$, and $\textsc{extra} =
1$, that is 128\,K slices per second: about 2\,MB/s in 1000 perf records, before
sampling. This is the quantity $p$ exists to control, and it is why sampling had
to be made safe rather than merely available.

\section{What \sysname Gets Wrong}
\label{sec:approx}

A tool that assigns blame will sometimes assign it wrongly, and an operator acting
on its output deserves to know how. This section lists every approximation the
design accepts, with the direction of the resulting error and the conditions under
which it appears. We treat the taxonomy as part of the contribution: it is what
makes the reported ratio interpretable, and several entries bound the settings in
which \sysname should be believed.

\subsection{Accounting}

\textbf{A1. Demand is not task-weighted.} Internal contention, external
contention, and throttling each contribute one demand unit per affected CPU,
regardless of how many of the target's tasks are queued there
(\S\ref{sec:model}). \sysname measures how much of the target's opportunity to run
was taken, not aggregate lost task time. For a target with a stable per-CPU
runnable count the two differ by a factor that cancels in the ratio. When
parallelism changes sharply within an interval they diverge, and the ratio
under-represents contention suffered by a deeply queued group relative to a
task-weighted counter.

\textbf{A2. Ending-boundary snapshots.} A slice is attributed using the waiting
state at its end, applied to its full duration. A target that wakes mid-slice is
charged for the whole slice; one that stops waiting mid-slice is charged for none
of it. The error is bounded by one slice duration per transition and is unbiased
in direction over many transitions. It grows with slice length, so it is largest
for CPU-bound competitors running long uninterrupted slices, which is also the
case where the ending state most likely held for the whole slice.

\textbf{A3. Unclamped anomaly threshold.} The threshold $K \cdot P_{99}$ is not
clamped to the ratio's maximum of $1$. If the baseline is high enough that the
threshold reaches $1$, no sample can exceed it and reporting is suppressed until
the baseline or $K$ falls. This silences alerts for persistently and severely
contended targets, arguably the ones an operator most wants to hear about, and it
is the entry we would fix first.

\subsection{Scheduler state}

\textbf{S1. Migration repairs only the source.} The destination bitmap is not
written remotely (Invariant~\ref{inv:migrate}). Between a migration and the
destination's next relevant event, a normal destination under-reports waiting,
while a stale bit at an already-throttled destination can over-attribute external
contention. The window closes at the destination's next switch, wakeup, or
unthrottle.

\textbf{S2. Cgroup task moves are invisible.} No hook observes a task moving
between CPU cgroups, so the source or destination group's bit can go stale. A
stale source bit persists until some later event references that task group, which
may never happen if the group goes idle.

\textbf{S3. Epoch-wide bitmap reset.} Publishing a new epoch clears every CPU's
bitmap (\S\ref{sec:bpf:epoch}). Targets unaffected by the reconfiguration
under-report until later observations rebuild their bits, bounded by the time to
the next scheduling event on each CPU.

\textbf{S4. Exact task-group scope.} A target is exactly the CPU cgroup named by
its \code{task\_group}; descendant cgroups are distinct task groups and are not
aggregated into it. Selection must name the cgroup whose tasks are to be measured.

\subsection{Identity}

\textbf{I1. Bounded target set.} At most $M$ containers are measured at once.
\sysname measures a sample of the machine's containers, not all of them.

\textbf{I2. Sparse CSS-ID limit.} A cgroup whose CSS ID exceeds 4095 does not fit
the 12-bit wire field. It is excluded from both sides of attribution, so it can
neither be blamed nor be measured, and it is counted.

\textbf{I3. Identity ordering and reuse.} A slice arriving before its competitor's
identity record is discarded. After a CSS ID is recycled, an already queued slice
can resolve to the new cgroup.

\textbf{I4. Immediate dense-slot reuse.} Slices or throttle records in flight when
a slot is reassigned may be attributed to the slot's new owner. This is the price
of reconfiguring without draining transport.

\textbf{I5. Lifecycle lag.} A delayed or lost lifecycle notification leaves target
classification stale until the next notification or the 60-second reconciliation.

\subsection{Sampling and transport}

\textbf{T1. Sampling variance.} Per Proposition~\ref{prop:variance}; short
contention events may fall entirely between samples.

\textbf{T2. Uncompensated transport loss.} Lost perf records and failed batch
submissions are not scaled away, and because loss correlates with load, sustained
loss biases estimates downward. Counters expose it.

\textbf{T3. Delayed partial batches.} A partial batch's age is checked only at the
next switch on that CPU, so if a CPU goes quiet its slices stay pending until a
later switch or shutdown. Slices carry duration but no timestamp, so userspace
attributes a delayed batch to the interval in which it arrives, and work can shift
between adjacent one-second intervals.

\textbf{T4. Duration saturation.} Slice and throttle durations above $2^{32}-1$\,ns
are saturated and counted, never wrapped.

\subsection{Throttling}

\textbf{H1. Direct throttling only.} \sysname observes throttling applied directly
to a target's own \code{cfs\_rq}. It does not propagate an ancestor's throttle
state into descendants, so time a target loses to an \emph{ancestor's} quota may
be attributed as external contention instead of $T$. This is the most consequential
misattribution in the list. On deployments that impose quota at a parent cgroup it
can blame neighbors for a limit the target's own hierarchy imposed, and it should
be read as a scoping condition on where \sysname's external ratio is trustworthy.

\textbf{H2. Interval attribution of throttle time.} A throttle interval is charged
entirely to the evaluation interval containing the unthrottle, rather than split
across the intervals it spans.

\textbf{H3. Attach during throttling.} If tracing starts mid-throttle, the first
throttle event includes the portion preceding attachment.

\textbf{H4. Throttle-event loss.} A lost throttle record undercounts $T$ in its
interval, which slightly inflates $r = E/D$, but does not corrupt waiting state or
later throttle accounting.

\section{Implementation}
\label{sec:impl}

\sysname is about 1{,}000 lines of C for the BPF object and 3{,}100 lines of Go for
the userspace tracer, excluding the container-discovery interface it shares with
the monitoring agent that hosts it. It runs on unmodified 4.18 (el8) and 5.10
kernels.

Supporting 4.18 shapes several choices. The BPF ring buffer is unavailable before
5.8, so transport uses a per-CPU perf event array. The verifier requires
compile-time-constant sizes for \code{bpf\_perf\_event\_output}, which is why
partial batches are emitted at power-of-two capacities (\S\ref{sec:bpf}).
Unthrottle handling uses an entry/return kprobe pair rather than a tracepoint,
since no tracepoint exposes a successful unthrottle together with its affected
runqueue.

The bitmap width $\textsc{extra}$, and therefore the target capacity
$M = 20 + 64\,\textsc{extra}$, is fixed at compile time for both artifacts because
it determines struct layouts on both sides of the wire. Userspace validates the BPF
map value size before attaching and validates the width carried in every batch
header, so a mismatched pair fails immediately instead of misattributing.
Everything else is runtime configuration: keep probability, batch size, perf ring
size and watermark, poll interval, the anomaly multiplier $K$, and the target
selection filters.

About a dozen counters are exported for self-diagnosis: perf submission failures
and the slices lost with them, records lost to a full ring, invalid runnable
counts, duration saturations, CSS-ID overflows, and upload drops, together with
peak occupancy of the perf ring and of the userspace queue. Transport loss is the
one error the estimator does not compensate (\S\ref{sec:analysis}), so making it
visible is part of the design rather than an afterthought.

\section{Preliminary Evaluation}
\label{sec:eval}

\sysname is deployed in production, and the measurements below are the ones we can
report with confidence today. They characterize cost, not accuracy or detection
quality; \S\ref{sec:eval:plan} states the design of the full evaluation. We label
this section preliminary rather than present a partial evaluation as a complete
one.

\subsection{Cost}

All figures come from a 96-core Intel Xeon production host tracking 84 real
production containers, with sampling disabled ($p = 1$, the most expensive
configuration).

\paragraph{Scheduler hooks.}
Under a Redis workload, attaching \sysname's hooks costs about 1\% of throughput,
measured by comparing runs with the hooks attached against runs with them
detached. We observed no significant change in p99 latency, which is the more
sensitive property for an interactive service: the hooks add bounded work inside an
already-held runqueue lock and introduce no new blocking.

\paragraph{Userspace tracer.}
The tracer consumes approximately 6\% of one core in the same configuration, as
reported by \code{top}, or about $0.06\%$ of the 96-core machine, including every
goroutine in the daemon.
This is an upper bound rather than a typical figure for two reasons. Sampling
scales the entire userspace path linearly in $p$, and attribution cost is
proportional to the number of \emph{waiting} targets per slice rather than to the
target capacity (\S\ref{sec:analysis}), so a machine with less contention pays
less.

\paragraph{Transport headroom.}
Over the same period, peak occupancy of the per-CPU perf rings was about 1\% of
capacity and peak occupancy of the shared userspace record queue about 15\%. These
are backpressure indicators, not costs. They say the kernel side ran with roughly
two orders of magnitude of headroom and the userspace side with a factor of six,
which matters because transport loss is the one error source the estimator does
not compensate (\S\ref{sec:analysis}, T2). The queue is the tighter of the two and
is the metric to watch as target counts or machine sizes grow.

\paragraph{Reading these numbers.}
The 84 targets are the compiled capacity at the default bitmap width, so this
configuration also exercises the widest snapshot and the largest charge matrix the
default build supports. Together the figures say that culprit attribution costs
roughly 1\% of application throughput and a fifteenth of a core per machine at full
sampling, which is what ``continuously on'' has to mean for a fleet.

\subsection{Operational experience}

\sysname's anomaly reports are used operationally, and have been most valuable for
bursty contention: the pattern on-demand tracing tends to miss, because by the time
an operator attaches a tracer the episode is over. We offer this as deployment
experience rather than as a measured result. Detection quality has not been
evaluated against a labelled set, and E4 below states how we intend to measure it.

\subsection{Planned evaluation}
\label{sec:eval:plan}

The following experiments need no kernel modification, so others can reproduce
them.

\paragraph{E1: Attribution accuracy against a constructed ground truth.}
Pin one victim and $n$ competing CPU hogs to a single CPU, making the culprit set
unambiguous by construction and each hog's true CPU consumption independently
readable from \code{cpuacct}. Vary the hogs' shares and duty cycles and compare
\sysname's charge-matrix row against those independent measurements, reporting
error both in the total and in the blame \emph{split} across competitors. Then
extend to multiple CPUs and multiple simultaneous victims, the regime where
correlational attribution is documented to become inconsistent~\cite{panda}, and
report whether \sysname's error is unchanged there, as the design predicts.

\paragraph{E2: Estimator behaviour versus sampling rate.}
Sweep $p$ from $1.0$ to $0.01$ under fixed contention and compare the observed
spread of $E_t$ against Proposition~\ref{prop:variance}, alongside tracer CPU cost.
The result is the accuracy/overhead curve an operator needs to choose $p$.

\paragraph{E3: Hook cost isolation.}
Measure the added cost per context switch on a switch-saturating microbenchmark,
separating always-on state maintenance from the sampling-dependent snapshot and
transport path, and report cost as a function of $p$ and of the number of waiting
targets.

\paragraph{E4: Detection quality.}
Inject contention of known magnitude, duration, and burstiness against a service
with a measurable latency SLO, and report detection latency, precision, and recall
of the $K \cdot P_{99}$ rule as a function of $K$, including the suppression regime
predicted by approximation A3.

\paragraph{E5: Scaling and reconfiguration.}
Measure cost as a function of target count $M$ and machine size, and the
measurement gap an epoch publication introduces (approximation S3) under container
churn.

\paragraph{E6: Comparison.}
Run \sysname alongside a preemption-tagging tracer of the Netflix
form~\cite{netflix} in a scenario where the victim is starved without being
preempted, that is, competitors run while the victim is never picked, to quantify
the blame that event tagging misses by construction. Compare separately against the PSL/PSP
decision matrix of Volpert et al.~\cite{volpert} across their eight contention
scenarios, checking that \sysname agrees on the quadrant while additionally naming
the source, and that its directly measured $I_t$ and $T_t$ track the self-disruption
cases their idle-task proxy identifies.

\section{Limitations and Future Work}
\label{sec:limitations}

\paragraph{Scope.}
\sysname measures CFS. Tasks in the real-time and deadline classes consume CPU
that a target waits for but are not represented, so on hosts running significant
real-time work the ratio under-attributes. Contention through shared caches,
memory bandwidth, and SMT siblings is outside the model as well: \sysname
attributes \emph{runqueue} contention, the component that is architecturally
attributable to an identifiable cgroup. Hardware-counter techniques address the
microarchitectural component~\cite{cpi2,bubbleup} and are complementary.

\paragraph{Cgroup v2.}
The design depends on \code{task\_group}, \code{h\_nr\_running}, and the CSS
identity, none of which is specific to the hierarchy version, so we expect v2 to
work. Only v1 has been tested. Validating v2, along with the unified hierarchy's
different container-to-cgroup layout, is the immediate next step.

\paragraph{Newer kernels.}
Deployment targets 4.18 and 5.10 today, and 7.x is planned. The ordering
invariants of \S\ref{sec:bpf:order} must be re-verified for each kernel, which is a
real maintenance obligation: the design substitutes documented ordering guarantees
for kernel patches, trading a porting cost for a verification cost. Newer kernels
also relax the constraints that shaped transport, since BPF ring buffers replace
the perf event array and the verifier accepts variable output sizes, so the
batching machinery can be simplified where those facilities exist.
\code{sched\_ext}~\cite{schedext} raises a further question, because a BPF
scheduler can change the semantics of the counters \sysname reads.

\paragraph{Hierarchy.}
A target is exactly one CPU cgroup (approximation S4), so a container whose tasks
live in descendant cgroups requires selecting those descendants. Aggregating a
subtree into one dense slot would generalize this, and the bitmap representation
admits it, but the task-group-to-dense-ID mapping would no longer be one-to-one.

\paragraph{Ancestor throttling.}
Approximation H1 is the correctness gap we would close first: time lost to an
ancestor's bandwidth limit is currently charged as external contention. Detecting
it requires observing the target's hierarchical throttle state rather than only its
own \code{cfs\_rq->throttled}, a bounded extension to the hook set.

\paragraph{Task weighting.}
The indicator formulation (approximation A1) is what makes a target's per-CPU state
one bit. Recovering task-weighted demand would require carrying counts instead of
bits, and the interesting question is whether a small saturating counter per
target, two or three bits, recovers most of the fidelity at acceptable width.

\paragraph{From attribution to action.}
\sysname produces a per-culprit signal but does not act on one. The signal is
directly usable by interference-aware schedulers and controllers that today must
infer culprits statistically~\cite{heracles,parties,quasar}. Closing that loop, by
feeding a blame matrix into placement or throttling decisions, is the natural next
system.

\section{Related Work}
\label{sec:related}

\paragraph{Victim-side kernel signals.}
PSI reports the fraction of wall time a cgroup's tasks were stalled on
CPU~\cite{psi}. Scheduler statistics accumulate per-entity run-queue
wait~\cite{schedstats}. The CPU controller reports quota-induced
throttling~\cite{cgroupv2,schedbwc}. These are the signals fleets actually deploy,
and they are the baseline \sysname extends: each quantifies the victim's suffering
accurately, and none names a cause. eBPF tools in the same family, \code{runqlat}
and \code{runqslower}~\cite{bcc,runqlat}, refine the distribution of that wait
without changing what it describes.

\paragraph{Scheduler-event instrumentation.}
The closest prior work instruments the scheduler itself. Volpert et al.\ hook
\code{sched\_wakeup} to timestamp a task becoming runnable and \code{sched\_switch}
to close out that wait when the task is scheduled, and aggregate the result into
two per-second metrics for an observed cgroup: average process scheduling latency
(PSL) and the average count of switches whose incoming task is the idle task
(PSP), the latter read as an indicator of quota throttling. Their decision matrix
then distinguishes an undisturbed workload, a noisy neighbor (high PSL, low PSP),
and disruption from inside the isolation group (high PSL, high PSP), online and
without workload profiles or offline analysis~\cite{volpert}.

That separation of external contention from self-inflicted disruption is the same
distinction \sysname draws, and achieving it without per-workload models is the
contribution. \sysname differs in two respects. First, it measures the quantities
the matrix infers: internal contention $I_t$ comes from a target running while its
own sibling waits, and throttling $T_t$ from the throttle and unthrottle
transitions themselves, rather than from idle-task switches used as a proxy. So
$I_t$, $E_t$, and $T_t$ are three separately measured terms rather than one
qualitative quadrant. Second, and more fundamentally, their instrumentation reads
only tasks of the observed cgroup. At each switch it either counts an idle-task
preemption or closes out an observed task's wait; the competing cgroup is never
examined. The method therefore establishes that a victim is being disturbed from
outside without identifying what is disturbing it, which is precisely the question
\sysname answers. Their evaluation covers eight contention scenarios on Kubernetes
with \code{stress-ng} and \code{nsjail}, and reports no overhead figures, noting
in its threats to validity that eBPF instrumentation may itself perturb the
measurement.

Netflix's tracer does attach an identity: it hooks \code{sched\_wakeup} and
\code{sched\_switch}, exports per-container run-queue latency, and tags each
preemption with the cgroup of the outgoing task, classified as same-container,
other-container, or system service~\cite{netflix}. \sysname differs in what is
attributed and in how much is attributable. Their unit is the \emph{preemption
event}; \sysname's is \emph{CPU time}, in a competitor $\times$ victim matrix
normalized by the victim's own demand, which yields a comparable ratio rather than
a count. More fundamentally, a preemption tag names the single cgroup that
displaced the task at one instant, whereas \sysname charges every cgroup that ran
during the whole interval in which the victim was waiting, including cgroups that
never preempted it and starvation episodes with no preemption at all. The
mechanisms differ correspondingly: a per-PID hash map of enqueue timestamps and
per-event map operations, against a fixed per-CPU bitmap with no per-task state;
and rate-limited event dropping, which is load-correlated and uncompensated,
against uniform sampling with inverse-probability rescaling (\S\ref{sec:analysis}).

\paragraph{Statistical and causal attribution.}
CPI\textsuperscript{2} identifies CPI outliers and nominates likely perpetrators
for throttling at Google scale~\cite{cpi2}. PANDA improves antagonist ranking with
global historical knowledge and a machine-level CPI metric, and documents the
family's central weakness: when multiple victims coexist and share contention,
correlation-based attribution becomes inconsistent~\cite{panda}. Recent work
applies Granger causality to resource time series to move from detection toward
causal attribution~\cite{causal}, and unsupervised methods detect noisy-neighbor
episodes without naming their source~\cite{unsupervised}. These methods run above
the kernel, scale across a cluster, and need no scheduler instrumentation.
\sysname sits at the opposite point. It observes the responsible party at the
scheduling event instead of inferring it from aggregate metrics, so multi-victim
coexistence, the documented failure mode, does not affect it; the cost is running
on every machine and measuring a bounded set of targets. The two compose well: a
cluster-level detector can select where to look, and \sysname can say who.

\paragraph{Interference-aware placement and isolation.}
Bubble-Up predicts collocation sensitivity by profiling~\cite{bubbleup}. Paragon
and Quasar classify and place workloads by interference
profile~\cite{paragon,quasar}. Heracles and PARTIES enforce QoS by dynamically
partitioning resources~\cite{heracles,parties}. Alioth monitors interference for
public-cloud multi-tenancy with learned models~\cite{alioth}. This line of work
consumes an interference signal and acts on it, generally deriving that signal from
profiling or from correlated metrics. \sysname does not schedule; it supplies the
per-culprit signal such systems currently lack.

\paragraph{Resource attribution more broadly.}
Retro attributes and controls resource usage across distributed request
flows~\cite{retro}, and continuous profilers attribute CPU \emph{consumption} to
code. Both answer ``who used the resource''. \sysname answers the adjacent and
harder question, ``whose use of the resource denied it to whom''. Lozi et
al.~\cite{wastedcores} make the complementary point that scheduler pathologies can
stay invisible to standard metrics for years, which argues for mechanism-level
instrumentation of the kind \sysname performs.

\paragraph{Sampling in tracers.}
Inverse-probability estimation is standard in survey sampling~\cite{ht1952} and
appears in profilers, but sampled tracers commonly sample the same event stream
they use to maintain state. That couples the sampling rate to correctness and
pushes implementations toward biased, load-dependent dropping. Separating state
maintenance from measurement (\S\ref{sec:bpf:sampling}) is what lets \sysname
sample uniformly and therefore estimate without bias.

\section{Conclusion}
\label{sec:conclusion}

Contention is a relationship, but every deployed CPU signal is a scalar attached to
a victim. \sysname closes that gap by inverting the accounting: instead of
measuring how long a container waited, it measures the CPU time every other
container consumed while that container was runnable and not running. Blame
becomes an observation at the scheduling event rather than an inference over
aggregate metrics.

Making the inversion deployable took three decisions that we believe generalize
beyond this tool. Splitting identity into a sparse space for parties that may be
blamed and a dense space for parties being measured turns a cross product into a
bitmap. Stamping each completed run slice with that bitmap makes records
self-describing, which makes userspace stateless and makes data loss cost
measurements rather than correctness. And separating state maintenance from
measurement makes sampling a pure variance-for-cost trade with no correctness
cliff, which is what lets the tool stay on continuously. Staying on is what catches
the bursty contention that on-demand tracing structurally cannot.

\sysname runs on unmodified 4.18 and 5.10 kernels and is deployed in production,
where culprit attribution costs about 1\% of Redis throughput and 6\% of one core
on a 96-core host tracking 84 containers. We have stated its error model in full,
because a tool that assigns blame owes its users the conditions under which it
assigns blame wrongly. The implementation will be released as open source.

\bibliographystyle{ACM-Reference-Format}
\bibliography{refs}

\end{document}